\documentclass[journal,letterpaper,twocolumn]{IEEEtran}

\usepackage{amsfonts}
\usepackage{amssymb}
\usepackage{cite}
\usepackage[cmex10]{amsmath}
\usepackage{float}
\usepackage{color}
\usepackage{stfloats}
\usepackage{mathrsfs}
\usepackage{algorithm}
\usepackage{algorithmic}
\usepackage{amsthm}
\usepackage{multirow}
\usepackage{changepage}
\usepackage[normalem]{ulem}
\usepackage{geometry}
\usepackage{caption}
\usepackage{booktabs}
\usepackage{multirow}
\usepackage{ragged2e}
\justifying
\newtheorem{lemma}{Lemma}

\IEEEoverridecommandlockouts

\ifCLASSINFOpdf
\usepackage[pdftex]{graphicx}
\DeclareGraphicsExtensions{.pdf,.jpeg,.png}
\else
\usepackage[dvips]{graphicx}
\DeclareGraphicsExtensions{.eps}
\fi
\usepackage{subfigure}
\usepackage{fancybox,dashbox}
\begin{document}
\title{Antenna Selection in MIMO Cognitive Radio-Inspired NOMA Systems}

\author{Yuehua~Yu,~He~Chen,~Yonghui~Li,~Zhiguo~Ding, and Li~Zhuo% <-this % stops a space
%\thanks{This work was supported in part by Australian Research Council grants FL160100032, DP150104019, and FT120100487. (\emph{Corresponding author: He Chen})}
\thanks{Y.~Yu, H.~Chen, and Y.~Li are with the School of Electrical and Information Engineering, University of Sydney, NSW 2006, Australia (email: \{yuehua.yu, he.chen, yonghui.li\}@sydney.edu.au). Z. Ding is with the School of Computing and Communications, Lancaster University, Lancaster LA1 4YW, U.K. (email: z.ding@lancaster.ac.uk). L. Zhuo is with the Signal and Information Processing Laboratory, Beijing University of Technology, Beijing, China. (email: zhuoli@bjut.edu.cn).}
}

\maketitle

\begin{abstract}
This letter investigates a joint antenna selection (AS) problem for a MIMO cognitive radio-inspired non-orthogonal multiple access (CR-NOMA) network. In particular, a new computationally efficient joint AS algorithm, namely subset-based joint AS (SJ-AS), is proposed to maximize the signal-to-noise ratio of the secondary user under the condition that the quality of service (QoS) of the primary user is satisfied. The asymptotic closed-form expression of the outage performance for SJ-AS is derived, and the minimal outage probability achieved by SJ-AS among all possible joint AS schemes is proved. The provided numerical results demonstrate the superior performance of the proposed scheme.
\end{abstract}
\begin{IEEEkeywords}
Non-orthogonal multiple access (NOMA), cognitive radio (CR), antenna selection
\end{IEEEkeywords}
\IEEEpeerreviewmaketitle

\section{Introduction}
Non-orthogonal multiple access (NOMA) and cognitive radio (CR) have emerged as efficient techniques to improve the spectral efficiency \cite{Ref_noma_5G,Ref_cognitive_radio}. By naturally combining the concepts of both NOMA and CR, a cognitive radio-inspired NOMA (CR-NOMA) scheme was proposed and studied in \cite{Ref_CR_NOMA}. In CR-NOMA, the unlicensed secondary users (SU) is opportunistically served under the condition that the quality of service (QoS) of the licensed primary users (PU) is satisfied. As a result, the transmit power allocated to the SU is constrained by the instantaneous signal-to-interference-plus-noise ratio (SINR) of the PU. Compared to the conventional CR systems, higher spectral efficiency can be achieved by CR-NOMA because both the PU and SU can be served simultaneously using the same spectrum.

Recently, multiple-input multiple-output (MIMO) techniques have been considered in CR-NOMA systems to exploit the spatial degrees of freedom \cite{Ref_CR_NOMA_MIMO}. To avoid high hardware costs and computational burden while preserving the diversity and throughput benefits from MIMO, the antenna selection (AS) problem for MIMO CR-NOMA systems has been investigated in \cite{Ref_ME2}, wherein the SU was assumed to be rate adaptive and the design criterion was to maximize the SU's rate subject to the QoS of PU. On the other hand, the outage probability has also been commonly used to quantify the performance of AS for an alternative scenario, wherein users have fixed transmission rates \cite{Ref_outage2}. To the best of the authors' knowledge, the outage-oriented AS schemes for CR-NOMA systems have not been studied in~open~literature.

Motivated by this, the design and analysis of the outage-oriented joint AS algorithm for MIMO CR-NOMA networks is studied in this letter, which is fundamentally different from that for orthogonal multiple access (OMA) networks. This is because there is severe inter-user interference in NOMA scenarios, wherein the signals are transmitted in an interference-free manner in OMA scenarios. Moreover, the transmit power allocated to the SU in CR-NOMA scenarios is constrained by the instantaneous SINR of the PU, which is affected by the antenna selection result. In this case, the joint antenna selection for NOMA networks is coupled with the power allocation design at the BS, which makes the design and analysis of the joint AS problem for CR-NOMA networks more challenging. In this letter, we propose a new low-complexity joint AS scheme, namely subset-based joint AS (SJ-AS), to maximize the signal-to-noise ratio (SNR) of the SU under the condition that the QoS of the PU is satisfied. The asymptotic closed-form expression of the outage performance for SJ-AS is derived, and the minimal outage probability achieved by SJ-AS among all possible joint AS schemes is proved. Numerical results demonstrate the superior performance of the proposed~scheme.\vspace{-0.5em}

\section{System Model and Proposed Joint AS Scheme}
Consider a MIMO CR-NOMA downlink scenario as \cite{Ref_CR_NOMA_MIMO},  wherein two users including one PU and one SU are paired in one group to perform NOMA. We consider that BS, PU and SU are equipped with $N$, $M$ and $K$ antennas, respectively. We assume that the channels between the BS and users undergo spatially independent flat Rayleigh fading, then the entries of the channel matrix, e.g., ${\tilde{h}_{nm}}$ (${\tilde{g}_{nk}}$), can be modelled as independent and identically distributed complex Gaussian random variables, where $\tilde{h}_{nm}$ ($\tilde{g}_{nk}$) represents the channel coefficient between the $n$th antenna of the BS and the $m$th ($k$th) antenna of the PU (SU). For notation simplicity, we define ${h}_{nm}=|\tilde{h}_{nm}|^2$ and ${g}_{nk}=|\tilde{g}_{nk}|^2$.
%\footnote{Noted that the two-user form of NOMA has been included in the 3GPP-LTE Advanced recently \cite{Ref_3GPP}.}

As in \cite{Ref_one_RF1}, we consider that the BS selects one (e.g., $n$th) out of its $N$ antennas to transmit information, while the users select one (e.g., $m$th and $k$th) out of $M$ and $K$ available antennas respectively to receive massages. In this sense, only one RF chain is needed at each node to reduce the hardware cost, power consumption and complexity, and only the partial channel state information, i.e., the channel amplitudes, is needed at the BS, which is assumed perfectly known at the BS through the control~signalling.

According to the principle of NOMA, the BS broadcasts the superimposed signals {\small$\sqrt{a}s_p + \sqrt{b}s_s$}, where $s_p$ ($s_s$) denotes the signal to the PU (SU) with {\small$\mathbb{E}[{\left| {{s_p}} \right|^2}]=\mathbb{E}[{\left| {{s_s}} \right|^2}]=1$}, and $a$ and $b$ are the power allocation coefficients satisfying $a+b=1$. Then the received signals at the PU and SU are given~by
\begin{small}\begin{eqnarray}\label{Equ_part1_y}
{y_p}&=&\sqrt{P}\tilde{h}_{nm}\left(\sqrt{a}s_p + \sqrt{b}s_s\right)+n_p,\\
{y_s}&=&\sqrt{P}\tilde{g}_{nk}\left(\sqrt{a}s_p + \sqrt{b}s_s\right)+n_s,
\end{eqnarray}\end{small}where $P$ is the transmit power at the BS, and $n_p$ ($n_s$) is the complex additive white Gaussian noise with variance {\small$\sigma_p^2$} ({\small$\sigma_s^2$}). For simplicity, we assume {\small$\sigma_p^2 = \sigma_s^2=\sigma^2$}.

Following the principle of CR-NOMA, $s_p$ is decoded by treating $s_s$ as noise at both users, and $s_s$ may be recovered at the SU when $s_p$ has been successfully subtracted in the SIC procedure. By denoting the transmit SNR as $\rho=P/\sigma^2$, the received SINR of decoding $s_p$ at the PU is given by
\begin{small}\begin{eqnarray}\label{Equ_Rp1}
\gamma_{p}\!\!&=&\!\!{ah_{nm}}/\left(bh_{nm}+1/\rho\right).
\end{eqnarray}\end{small}Similarly, the received SINR to detect $s_p$ at the SU is given~by
\begin{small}\begin{eqnarray}\label{Equ_Rp2}
\gamma_{s\rightarrow p}\!\!&=&\!\!{ag_{nk}}/\left(bg_{nk}+1/\rho\right).
\end{eqnarray}\end{small}When $s_p$ is successfully removed, the SNR to detect $s_s$ at the SU is given by\vspace{-0.5em}
\begin{small}\begin{eqnarray}\label{Equ_Rs}
\gamma_s&=&bg_{nk}\rho.
\end{eqnarray}\end{small}\vspace{-1.5em}

Let {\small${{\gamma}}^{th}_p$} ({\small${\gamma}^{th}_s$}) denotes the predetermined detecting threshold of $s_p$ ($s_s$). As the SU is served on the condition that {\small${\gamma}^{th}_p$} is met, mathematically, $\gamma_{p}$ and $\gamma_{s\rightarrow p}$ should satisfy the following constraint simultaneously: {\small$\min\left(\gamma_{p},\gamma_{s\rightarrow p}\right)\geqslant \gamma^{th}_p$.}\vspace{-1em}

\subsection{The Formulation of Joint AS Optimization Problem}
In order to maximize the received SNR of the SU, we would like to solve the following optimization~problem:
\begin{small}\begin{subequations}\label{Equ_problem}
\begin{align}
      \mathbf{P}:~\{b^*, n^*,m^*,k^*\}=&\mathop {\arg \max }\limits_{b,n \in {\mathcal{N}},m \in {\mathcal{M}},k \in {\mathcal{K}}} {\gamma_s}\left(b,{g}_{nk}\right),\\
      \label{Equ_problem_c1}
      &\mathrm{s.t.}~\min\left(\gamma_{p},\gamma_{s\rightarrow p}\right)\geqslant\gamma^{th}_p,\\
      \label{Equ_problem_c2}
      &~~~~~0\leq b<1.
      \end{align}
\end{subequations}\end{small}where {\small$\mathcal{N}\!\!=\!\!\{1,\cdots,N\}, \mathcal{M}\!\!=\!\!\{1,\cdots,M\}, \small\mathrm{and}~\mathcal{K}\!\!=\!\!\{1,\cdots,K\}$}, and $\mathbf{P}$ is the joint optimization problem of antenna selection and power allocation. Specifically, similar to \cite{Ref_power_allocation}, given the antenna indexes $n$, $m$ and $k$, the optimal power allocation strategy $b$ can be obtained based on Lemma 1.
\begin{lemma}\label{Lemma:optimal_b}
Given the antenna indexes $n$, $m$ and $k$, the optimal power allocation strategy $b$ is given~by
\begin{small}\begin{eqnarray}\label{Equ_b2}
b=\max\left(\left(\beta\rho-{\gamma}^{th}_p\right)/\left(\left({\gamma}^{th}_p+1\right)\beta\rho\right),0\right),%\overset{\rho\rightarrow\infty}{\approx}\frac{\beta\rho}{\left(\gamma_p+1\right)\beta\rho}.
\end{eqnarray}\end{small}where $\beta=\min\left(h_{nm},~g_{nk}\right)$.
\end{lemma}
\begin{proof}
Given antenna indexes $n$, $m$ and $k$, by substituting (\ref{Equ_Rp1})-(\ref{Equ_Rp2}) into (\ref{Equ_problem_c1}), the power coefficient $b$ should satisfy the condition: {\small{$b\!\leqslant\!\frac{\beta\rho-{\gamma}^{th}_p}{\left({\gamma}^{th}_p+1\right)\beta\rho}$}}. Meanwhile, $\gamma_s$ is an increasing function of $b$  as shown in (\ref{Equ_Rs}). In this case, in order to maximize $\gamma_s$,  $b$ should take the maximum value in its range. By noting that $0\leqslant b<1$, we then can express the optimal power allocation coefficient $b$ as in (\ref{Equ_b2}). The proof is completed.
\end{proof}

By substituting (\ref{Equ_b2}) into (\ref{Equ_Rs}) and when $b>0$, we have \begin{small}\begin{eqnarray}\label{Equ_Rs2}
\gamma_s\left(h_{nm},g_{nk}\right)=\frac{\min\left(h_{nm},g_{nk}\right)\rho-{\gamma}^{th}_p}{\left({\gamma}^{th}_p+1\right)\min\left(h_{nm},g_{nk}\right)} g_{nk},
\end{eqnarray}\end{small}otherwise $\gamma_s\!=\!0$. At this point, the joint optimization problem $\mathbf{P}$ is simplified into the joint antenna selection problem. It is straightforward to see that finding the optimal antenna indexes $\{n^*,m^*,k^*\}$ may require an exhaustive search (ES) over all possible antenna combinations with the complexity of \footnote{$\mathcal{O}$ is usually used in the efficiency analysis of algorithms and $q(x)=\mathcal{O}\left(p(x)\right)$ when $\lim\limits_{x\rightarrow\infty}|\frac{q(x)}{p(x)}|=c, 0<c<\infty$.}$\mathcal{O}\left(NMK\right)$. When $N$, $M$ and $K$ become large, the computational burden of ES may become unaffordable. Motivated by this, an computationally efficient joint AS algorithm for MIMO CR-NOMA systems will be developed in the next~subsection.\vspace{-1em}

\subsection{Proposed Subset-based Joint AS (SJ-AS) Scheme}
The aim of SJ-AS algorithm is to decrease the computational complexity by greatly reducing the searching set, while ensuring the QoS of the PU and maximizing the achievable SNR of the SU. Specifically, SJ-AS mainly consists of the following three stages.
\begin{itemize}
  \item{\textbf{Stage 1}}. Build the subset $\mathcal{S}_1=\big\{\left(h^{(n)},g^{(n)}\right),n\in\mathcal{N}\big\}$ to reduce the search space, where $h^{(n)}$ and $g^{(n)}$ are the maximum-value elements in the $n$th row of $\bf{H}$ and $\bf{G}$, respectively. Mathematically, we have
      \begin{small}\begin{eqnarray}\label{Equ_hn_gn_max}
      h^{(n)}&=&\max\left(h_{n1},\cdots,h_{nM}\right),\\
      g^{(n)}&=&\max\left(g_{n1},\cdots,g_{nK}\right).
      \end{eqnarray}\end{small}
  \item{\textbf{Stage 2}}. Build the subset $\mathcal{S}_2$ by selecting the pairs from $\mathcal{S}_1$, in which each pair ensures the target SINR of the PU can be supported and $s_p$ can be subtracted successfully at the SU. That~is,
  \begin{small}\begin{eqnarray}\label{Equ_S2}
  \mathcal{S}_2=\left\{\min\left(\gamma_{p}^{(n)},\gamma_{s\rightarrow p}^{(n)}\right)\geqslant\gamma^{th}_p,~n\in\mathcal{S}_1\right\},
  \end{eqnarray}\end{small}where $\gamma_{p}^{(n)}$ and $\gamma_{s\rightarrow p}^{(n)}$ can be obtained by substituting $h^{(n)}$ and $g^{(n)}$ into (\ref{Equ_Rp1}) and (\ref{Equ_Rp2}), respectively. Specifically, $b^{(n)}$ is given in (\ref{Equ_b2}) with $\beta^{(n)}=\min(h^{(n)},g^{(n)})$.
  \item{\textbf{Stage 3}}. When $|\mathcal{S}_2|>0$, select the antenna triple which can maximize the SNR for the SU, i.e.,
  \begin{small}\begin{eqnarray}\label{Equ_S3}
  \{n^*, m^*, k^*\}=\arg\max\left\{\gamma_s\left(h^{(n)},g^{(n)}\right), n\in\mathcal{S}_2\right\}.
  \end{eqnarray}\end{small}Let $m^*$ and $k^*$ denote the original column indexes of $h^{(n^*)}$ and $g^{(n^*)}$, respectively. That is, the $n^*$th antenna at the BS, and the $m^*$th and $k^*$th antennas at the PU and SU are jointly selected. In contrast, when $|\mathcal{S}_2|=0$, the system suffers from an outage.
\end{itemize}
As mentioned before, the complexity of the ES-based scheme is as high as $\mathcal{O}\left(NMK\right)$. In contrast, the complexity of the proposed SJ-AS scheme is upper bounded by $\mathcal{O}\left(N\left(M+K+2\right)\right)$. For the case $N\!=\!M\!=\!K$, we can find that the complexity of SJ-AS is approximately $\mathcal{O}\left(N^2\right)$, which is an order of magnitude lower than $\mathcal{O}\left(N^3\right)$ of the optimal ES-based scheme.

\section{Performance Evaluation}
In this section, we will analyse the system outage performance achieved by SJ-AS. By using the assumption that channel coefficients are Rayleigh distributed, the cumulative density functions (CDF) and the probability density functions (PDF) of $h^{(n)}$ and $g^{(n)}$ in $\mathcal{S}_1$ can be expressed as in \cite{Ref_ME2},\vspace{-0.5em}
\begin{small}\begin{eqnarray}\label{Equ_hi_CDF_PDF}
F_{h^{(n)}}(x)\!\!\!\!&=&\!\!\!\!\left(1-e^{-\Omega_hx}\right)^M, ~F_{g^{(n)}}(x)\!=\!\left(1-e^{-\Omega_gx}\right)^K,~\\
f_{h^{(n)}}(x)\!\!\!\!&=&\!\!-\sum\nolimits_{m=0}^M(-1)^m\binom{M}{m}m\Omega_he^{-m\Omega_hx},\\
f_{g^{(n)}}(x)\!\!\!\!&=&\!\!\!\!-\sum\nolimits_{k=0}^K(-1)^k\binom{K}{k}k\Omega_ge^{-k\Omega_gx},
\end{eqnarray}\end{small}where {\small$\Omega_h=1/\mathbb{E}\left[h_{im}\right]$}, {\small$\Omega_g=1/\mathbb{E}\left[g_{nk}\right]$}, and {\small$f_{h^{(n)}}(x)$} and {\small$f_{g^{(n)}}(x)$} are expanded based on the binomial theorem.

Let $\mathcal{O}_1$ denote the event $|\mathcal{S}_2|=0$, and $\mathcal{O}_2$ denote the event {\small$\gamma_s^{(n^*)}<{\gamma}^{th}_s$} while $|\mathcal{S}_2|>0$. As in \cite{Ref_outage}, the overall system outage is defined as the event that any user in the system cannot achieve reliable~detection, i.e.,\vspace{-0.2em}
\begin{small}\begin{eqnarray}\label{Equ_OverallPr}
\mathrm{Pr}\left(\mathcal{O}\right)={\mathrm{Pr}\left(\mathcal{O}_1\right)}+{\mathrm{Pr}\left(\mathcal{O}_2\right)}.
\end{eqnarray}\vspace{-0.2em}\end{small}In this case, the asymptotic system outage probability can be obtained according to the following lemma.
\begin{lemma}When the transmit SNR $\rho\rightarrow\infty$, the system outage probability achieved by SJ-AS can be approximated as
\begin{small}\begin{eqnarray}\label{lemma1}
\mathrm{P}(\mathcal{O})\!\!&\approx&\!\!\sum_{\ell=0}^N\binom{N}{\ell}\left(\sum_{m=1}^M\sum_{k=1}^Kc_{m,k}\left(e^{-\varphi_{m,k}{c_1}}-e^{-\phi_{m,k}}\right.\right.\nonumber\\
\!\!\!\!&+&\!\!\!\!\left.\left.\frac{k\Omega_ge^{-\varphi_{m,k}c_2}\left(1-e^{-\varphi_{m,k}c_1}\right)}{\varphi_{m,k}}\right)\right)^{\ell}F_{\beta^{(n)}}\!\!\left(c_1\right)^{N-\ell}.~~~
\end{eqnarray}\end{small}where {\small$c_{m,k}=(-1)^{m+k}\binom{M}{m}\binom{K}{k}$}, {\small$\varphi_{m,k}=m\Omega_h+k\Omega_g$}, {\small$c_1=\frac{{\gamma}^{th}_p}{\rho}$}, {\small$c_2={{\gamma}^{th}_s\left({\gamma}^{th}_p+1\right)}/{\rho}$}, {\small$\phi_{m,k}=m\Omega_hc_1+k\Omega_gc_2$}, and {\small$ F_{\beta^{(n)}}(x)=1-\left(F_{h^{(n)}}(x)-1\right)\left(F_{g^{(n)}}(x)-1\right)$}.
\end{lemma}
\begin{proof}
we can first calculate the term $\mathrm{Pr}\left(\mathcal{O}_1\right)$ as \vspace{-0.5em}
\begin{small}\begin{eqnarray}\label{Equ_O1}
\mathrm{Pr}\left(\mathcal{O}_1\right)\!\!\!\!\!&=&\!\!\!\!\!\mathrm{Pr}\left(|\mathcal{S}_2|=0\right)\!=\!\prod_{n=1}^{N}\mathrm{Pr}\left(\min\left(\gamma_{p}^{(n)},\gamma_{s\rightarrow p}^{(n)}\right)<{\gamma}^{th}_p\right)\nonumber\\
\!\!\!\!\!&=&\!\!\!\!\!\prod_{n=1}^{N}\!\!\mathrm{Pr}\!\left(b^{(n)}\leq0\right)\!=\!\prod_{n=1}^{N}\!\!\mathrm{Pr}\!\left(\beta^{(n)}\leq\frac{{\gamma}^{th}_p}{\rho}\right)\nonumber\\
\!\!\!\!\!&=&\!\!\!\!\!\left(F_{\beta^{(n)}}\left(c_1\right)\right)^{N}.~~~
\end{eqnarray}\end{small}where $c_1={{\gamma}^{th}_p}/{\rho}$ and the CDF of $\beta^{(n)}$ is given by \vspace{-0.5em}
\begin{small}\begin{eqnarray}\label{Equ_betaCDF}
F_{\beta^{(n)}}(x)\!\!\!&=&\!\!\!1-\mathrm{Pr}\left(\beta^{(n)}>x\right)\nonumber\\
\!\!\!&=&\!\!\!1-{\mathrm{Pr}\left(h^{(n)}>g^{(n)}> x\right)}-{\mathrm{Pr}\left(g^{(n)}>h^{(n)}> x\right)}\nonumber\\
\!\!\!&=&\!\!\!1-\left(F_{h^{(n)}}(x)-1\right)\left(F_{g^{(n)}}(x)-1\right).
\end{eqnarray}\end{small}By substituting (\ref{Equ_betaCDF}) into (\ref{Equ_O1}), $\mathrm{Pr}\left(\mathcal{O}_1\right)$ is~obtained.

We then turn to the calculation of $\mathrm{Pr}\left(\mathcal{O}_2\right)$,
\begin{small}\begin{eqnarray}\label{Equ_O2_1}
\mathrm{Pr}\left(\mathcal{O}_2\right)&=&\mathrm{Pr}\left(\rho g^{(n^*)}b^{(n^*)}<{\gamma}^{th}_s,~|\mathcal{S}_2|>0\right)\nonumber\\
&=&\mathrm{Pr}\left(g^{(n^*)}b^{(n^*)}<\frac{{\gamma}^{th}_s}{\rho},~|\mathcal{S}_2|>0\right).
\end{eqnarray}\end{small}Let $\alpha^{(n)}=g^{(n)}b^{(n)}$ for $\forall n\in \mathcal{S}_2$. Since $b^{(n)}>0$ for $\forall n\in \mathcal{S}_2$, the product $\alpha^{(n^*)}=g^{(n^*)}b^{(n^*)}$ in (\ref{Equ_O2_1}) can be presented~as\vspace{-0.5em}
\begin{small}\begin{eqnarray}
\alpha^{(n^*)}=\max\left(\alpha^{(n)}\right),~~\mathrm{for}~~\forall n\in \mathcal{S}_2.
\end{eqnarray}\end{small}Then $\mathrm{Pr}\left(\mathcal{O}_2\right)$ can be further expressed as\vspace{-0.5em}
\begin{small}\begin{eqnarray}\label{Equ_O2_2}
\mathrm{Pr}\left(\mathcal{O}_2\right)\!\!\!&=&\!\!\!\mathrm{Pr}\left(\alpha^{(n^*)}<\frac{{\gamma}^{th}_s}{\rho},~|\mathcal{S}_2|>0\right)\nonumber\\
\!\!\!&=&\!\!\!\sum_{\ell=1}^{N}\mathrm{Pr}\left(\alpha^{(n^*)}<\frac{{\gamma}^{th}_s}{\rho}\mid|\mathcal{S}_2|=\ell\right)\mathrm{Pr}\left(|\mathcal{S}_2|=\ell\right)\nonumber\\
\!\!\!&=&\!\!\!\sum_{\ell=1}^{N}\left(\mathrm{Pr}\left(\alpha^{(n)}<\frac{{\gamma}^{th}_s}{\rho}\mid|\mathcal{S}_2|=\ell\right)\right)^{\ell}\mathrm{Pr}\left(|\mathcal{S}_2|=\ell\right)\nonumber\\
\!\!\!&=&\!\!\!\sum_{\ell=1}^{N}\left(F_{\alpha^{(n)}}\left(\frac{{\gamma}^{th}_s}{\rho}\right)\right)^{\ell}\mathrm{Pr}\left(|\mathcal{S}_2|=\ell\right),
\end{eqnarray}\end{small}in which,\vspace{-1em}
\begin{small}\begin{eqnarray}\label{Equ_F_alpha}
F_{\alpha^{(n)}}(\frac{{\gamma}^{th}_s}{\rho})\!\!&=&\!\!\mathrm{Pr}\left(\frac{g^{(n)}\left(\beta^{(n)}-\frac{{\gamma}^{th}_p}{\rho}\right)}{({\gamma}^{th}_p+1)\beta^{(n)}}<\frac{{\gamma}^{th}_s}{\rho}\mid n\in\mathcal{S}_2, |\mathcal{S}_2|>0\right)\nonumber\\
\!\!&=&\!\!\frac{{\mathrm{Pr}\left(\frac{g^{(n)}\rho-{{\gamma}^{th}_p}}{{\gamma}^{th}_p+1}<{\gamma}^{th}_s,h^{(n)}\geqslant g^{(n)}>\frac{{\gamma}^{th}_p}{\rho}\right)}}{\mathrm{Pr}\left(\beta^{(n)}>\frac{{\gamma}^{th}_p}{\rho}\right)}\nonumber\\
\!\!&+&\!\!\frac{{\mathrm{Pr}\left(\frac{g^{(n)}\left(h^{(n)}\rho-{{\gamma}^{th}_p}\right)}{({\gamma}^{th}_p+1)h^{(n)}}<{\gamma}^{th}_s,\frac{{\gamma}^{th}_p}{\rho}<h^{(n)}<g^{(n)}\right)}}{\mathrm{Pr}\left(\beta^{(n)}>\frac{{\gamma}^{th}_p}{\rho}\right)}\nonumber\\
\!\!&=&\!\!\left(Q_1+Q_2\right)/\left(1-F_{\beta^{(n)}}(c_1)\right),
\end{eqnarray}\end{small}Where\vspace{-1em}
\begin{small}\begin{eqnarray}
Q_1\!\!&=&\!\!\mathrm{Pr}\left(\frac{{\gamma}^{th}_p}{\rho}<g^{(n)}<\frac{{\gamma}^{th}_s\left({\gamma}^{th}_p+1\right)+{\gamma}^{th}_p}{\rho}, h^{(n)}\geqslant g^{(n)}\right),~~~\\
Q_2\!\!&=&\!\!\mathrm{Pr}\left(\frac{{\gamma}^{th}_p}{\rho}<h^{(n)}<g^{(n)}<\frac{{\gamma}^{th}_s({\gamma}^{th}_p+1)h^{(n)}}{h^{(n)}\rho-{\gamma}^{th}_p}\right).
\end{eqnarray}\end{small}Let {\small$c_2={{\gamma}^{th}_s\left({\gamma}^{th}_p+1\right)}/{\rho}$}, {\small$c_{m,k}=(-1)^{m+k}\binom{M}{m}\binom{K}{k}$}, and {\small$\varphi_{m,k}=m\Omega_h+k\Omega_g$}, then $Q_1$ can be obtained as follows,
\begin{small}\begin{eqnarray}\label{Equ_Q1}
Q_1\!\!&=&\!\!\mathrm{Pr}\left(c_1<g^{(n)}<c_1+c_2, h^{(n)}\geqslant g^{(n)}\right)\nonumber\\
\!\!&=&\!\!\sum_{m=1}^M\sum_{k=1}^Kc_{m,k}\frac{k\Omega_ge^{-\varphi_{m,k}{c_1}}\left(1-e^{-\varphi_{m,k}{c_2}}\right)}{\varphi_{m,k}}.
\end{eqnarray}\end{small}Similarly, when $\rho\rightarrow\infty$, $Q_2$ can be approximated as\vspace{-0.5em}
\begin{small}\begin{eqnarray}\label{Equ_Q2}
Q_2\!\!\!\!&\approx&\!\!\!\!\mathrm{Pr}\left(c_1<h^{(n)}<g^{(n)}<c_2\right)\nonumber\\
\!\!\!\!&=&\!\!\!\!\sum_{m=1}^M\!\sum_{k=1}^K\!c_{m,k}\!\left(\!\frac{m\Omega_he^{-\varphi_{m,k}{c_1}}\!+\!k\Omega_ge^{-\varphi_{m,k}c_2}}{\varphi_{m,k}}\!-\!e^{-\phi_{m,k}}\!\right),~~~~
\end{eqnarray}\end{small}where $\phi_{m,k}=m\Omega_hc_1+k\Omega_gc_2$.

On the other hand, the probability that $|\mathcal{S}_2|=\ell$ can be calculated as below,\vspace{-1em}
\begin{small}\begin{eqnarray}\label{Equ_Sn}
\mathrm{Pr}\left(\mid\mathcal{S}_2\mid=\ell\right)\!\!\!\!&=&\!\!\!\!\binom{N}{\ell}\left(F_{\beta^{(n)}}\!\left(c_1\right)\right)^{N-\ell}\left(1\!\!-\!\!F_{\beta^{(n)}}\left(c_1\right)\right)\!^{\ell}.~~~~
\end{eqnarray}\end{small}
By combining (\ref{Equ_O2_2})-(\ref{Equ_Sn}) and applying some algebraic manipulations, $\mathrm{P}(\mathcal{O}_2)$ can be expressed as
\begin{small}\begin{eqnarray}\label{Equ_O2_3}
\mathrm{P}(\mathcal{O}_2)\!\!&\approx&\!\!\sum_{\ell=1}^N\binom{N}{\ell}\left(\sum_{m=1}^M\sum_{k=1}^Kc_{m,k}\left(e^{-\varphi_{m,k}{c_1}}-e^{-\phi_{m,k}}\right.\right.\nonumber\\
\!\!\!\!&+&\!\!\!\!\left.\left.\frac{k\Omega_ge^{-\varphi_{m,k}c_2}\left(1-e^{-\varphi_{m,k}c_1}\right)}{\varphi_{m,k}}\right)\right)^{\ell}F_{\beta^{(n)}}\!\!\left(c_1\right)^{N-\ell}.~~~
\end{eqnarray}\end{small}By summing (\ref{Equ_O1}) and (\ref{Equ_O2_3}), the proof of (\ref{lemma1}) is completed.
\end{proof}

\textit{Remark 1}: When $\rho$ approaches infinity, $c_1$, $c_2$ and $\phi_{m,k}$ approach zero. By using the binomial theorem and the approximation {\small$1\!-\!e^{-x}\!\overset{x\rightarrow0}{\approx}\!x$}, the system outage probability can be further approximated as follows:
\begin{small}\begin{eqnarray}\label{Equ_O_app}
\mathrm{P}(\mathcal{O})\!\!&\approx&\!\!\left(1-\left(1-F_{h^{(n)}}(c_1)\right)\left(1-F_{g^{(n)}}(c_1)\right)\right)^N\nonumber\\
%&=&\left(F_{h^{(n)}}(c_1)+F_{g^{(n)}}(c_1)-F_{h^{(n)}}(c_1)F_{g^{(n)}}(c_1)\right)^N\nonumber\\
&\approx&{\zeta^N}/{\rho^{N\min(M,K)}},
\end{eqnarray}\end{small}where {\small$\zeta\!\!=\!\!\frac{{\tilde{c}_1}^M}{\rho^{M-\min(M,K)}}\!\!+\!\!\frac{{\tilde{c}_2}^K}{\rho^{K-\min(M,K)}}\!\!-\!\!\frac{{\tilde{c}_1}^M{\tilde{c}_2}^K}{\rho^{M+K-\min(M,K)}}$, $\tilde{c}_1=\Omega_h{\gamma}^{th}_p$}, and {\small$\tilde{c}_2=\Omega_g{\gamma}^{th}_p$}. From (\ref{Equ_O_app}), we can see that the SJ-AS scheme can realize a diversity of $N\min\left(M,K\right)$.

\textit{Remark 2}: The optimality of the proposed SJ-AS is illustrated in the following lemma.

\begin{lemma}The proposed SJ-AS scheme minimizes the system outage probability of the considered MIMO CR-NOMA system.
\end{lemma}
\begin{proof}
This lemma can be proved by contradiction. Suppose there exists another joint AS strategy achieving a lower system outage probability than SJ-AS. Let {\small$\left(\hat{n}^*, \hat{m}^*,\hat{k}^*\right)\neq\left({n}^*, {m}^*,{k}^*\right)$} denote the antenna triple selected by the new strategy. According to the assumption, it is possible that there is no outage with {\small$\left(\hat{n}^*, \hat{m}^*,\hat{k}^*\right)$} antennas while an outage occurs with {\small$\left({n}^*, {m}^*,{k}^*\right)$} antennas. In this case, the pair {\small$\left(h_{\hat{n}^*\hat{m}^*},g_{\hat{n}^*\hat{k}^*}\right)$} must be in {\small$|\mathcal{S}_2|$} to satisfy the target SINR of the PU, i.e., {\small$|\mathcal{S}_2|>0$}. Recall that the pair {\small$\left(h_{{n}^*{m}^*},g_{{n}^*{k}^*}\right)\in|\mathcal{S}_2|$} is selected according to (\ref{Equ_S3})  to maximize {\small${\gamma}_s$}. In this case, if the maximized {\small$\gamma_s$} cannot meet {\small$\gamma_s^{th}$}, the antennas selected by other scheme which provides smaller {\small${\gamma}_s$} cannot meet {\small$\gamma_s^{th}$}, either. Therefore, one can conclude that there is NO other joint AS strategies can achieve a lower outage probability than SJ-AS, which is contradicted to the assumption made earlier. The lemma is~proved.
\end{proof}\vspace{-0.5em}

\section{Numerical Studies}
In this section, the performance of the proposed SJ-AS algorithm for MIMO CR-NOMA networks is evaluated by Monte Carlo simulations. Let $\Omega_h=d_p^\varepsilon$ ($\Omega_g=d_s^\varepsilon$), where $d_p$ ($d_s$) is the distance between the BS and PU (SU), and the path-loss exponent is set as $\varepsilon=3$.

Fig. (a) and (b) compare the received SNR of the SU and the system outage performance between SJ-AS and other AS strategies. As illustrated in both figures, over the entire SNR region, SJ-AS outperforms the conventional max-min scheme, in which the antenna selection is executed under the max-min criteria, i.e., {\small$\max(\min(h_{nm}, g_{nk}))$ for all $n\in\mathcal{N}, m\in\mathcal{M}$ and $k\in\mathcal{K}$}. Furthermore, the performance of both SJ-AS and the max-min scheme are much better than that of random AS, since both SJ-AS and the max-min scheme utilize the spatial degrees of freedom brought by the multiple antennas at each node. We also see that the analytical results match the simulation results for SJ-AS, which validates our theoretical analysis in Sec. III. Moreover, compared to the optimal ES scheme, SJ-AS can achieve the optimal outage performance as discussed in Remark 2, but with significantly reduced computational complexity. {In particular, the corresponding average power allocation coefficient $b$ for each scheme is illustrate in Table. I. Again we can find that the SJ-AS can achieve the same power allocation of the optimal ES scheme.}

\begin{figure}
\begin{minipage}[t]{0.5\linewidth}
\centering
\subfigure[Rx SNR vs. transmit power]{
\includegraphics[width=1.9in]{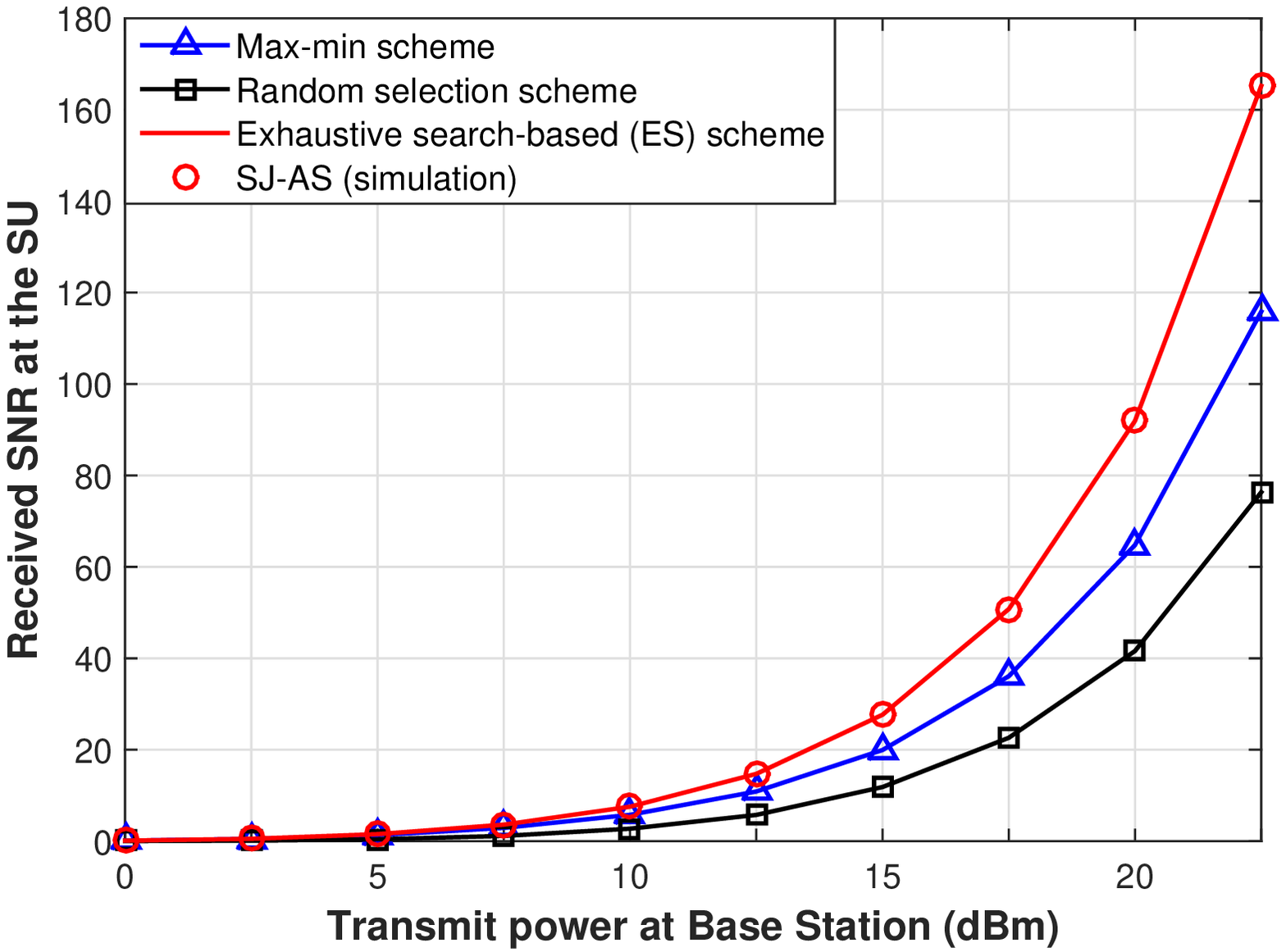}
\label{fig:side:a}}
\end{minipage}%
\begin{minipage}[t]{0.5\linewidth}
\centering
\subfigure[Outage probability vs. transmit power]{
\includegraphics[width=1.9in]{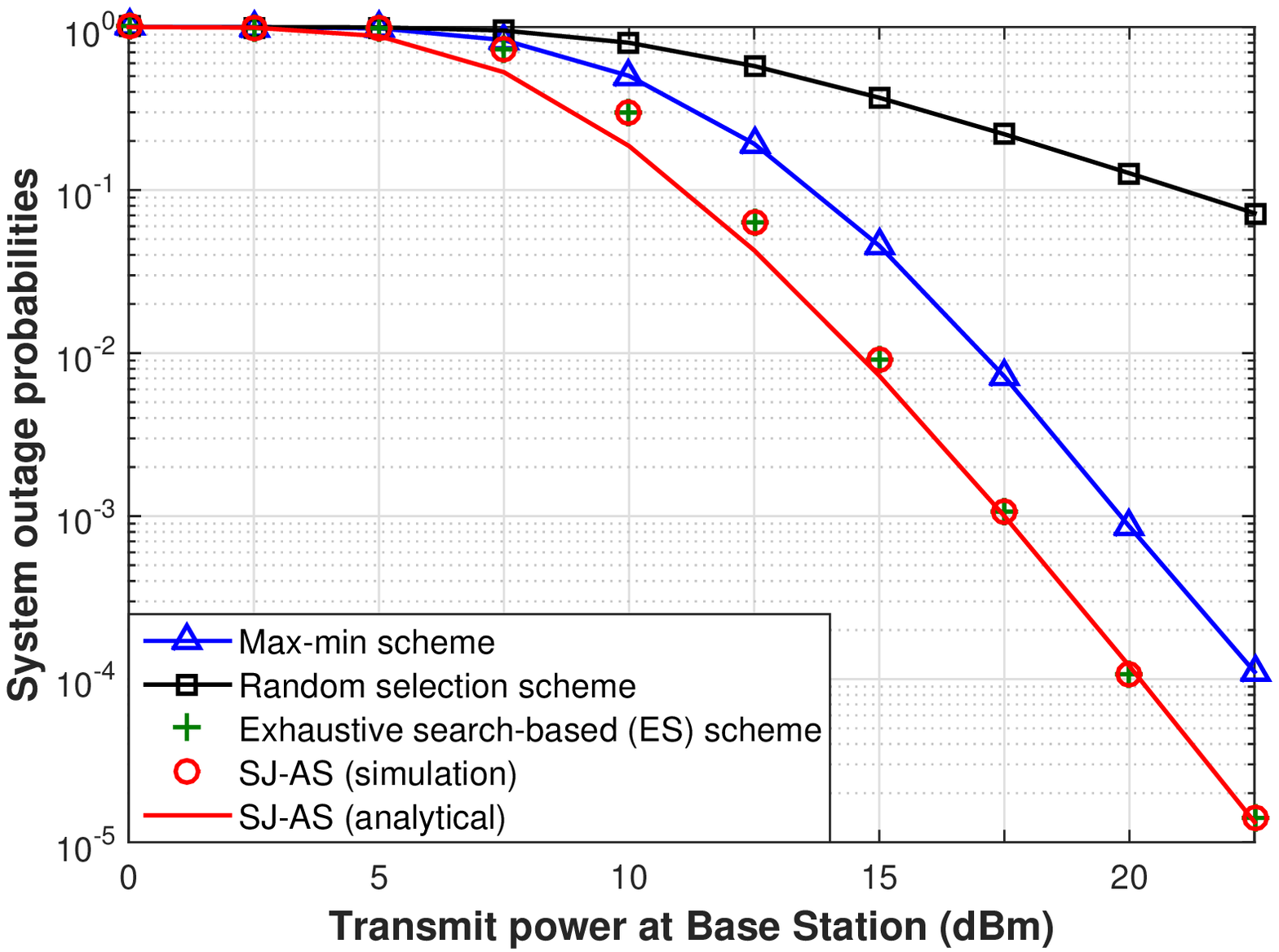}
\label{fig:side:b}}
\end{minipage}
\caption{{\small{Comparison of AS schemes, $N\!=\!M\!=\!K\!=\!2$, $d_p\!=\!350$m, $d_s\!=\!250$m, $\sigma\!=\!-70$dBm, $\gamma_p\!=\!2^{0.5}-1$ and $\gamma_s\!=\!2^{2.5}-1$}. }}
\end{figure}

\begin{footnotesize}
\begin{table}
\centering
\caption{\small{Average power allocation coefficient $b$.}}
\begin{tabular}{l*{4}{c}c}
Transmit power $P$          & $0$dBm &$5$dBm &$10$dBm & $15$dBm &$20$dBm \\
\hline
Random AS & 0.0155 & 0.1491 & 0.3715 & 0.5425 & 0.6359\\
Max-min AS& 0.1441&0.5006&0.6417&0.6864&0.7006 \\
ES AS         & 0.1418&0.4624 &0.5997&0.6641&0.6915\\
SJ-AS     &0.1418&0.4624 &0.5997&0.6641&0.6915\\
\end{tabular}\vspace{-2em}
\end{table}
\end{footnotesize}\vspace{-0.5em}

\section{Conclusion}
In this letter, we studied the joint AS and power allocation problem for a MIMO CR-NOMA system. A computationally efficient SJ-AS scheme was proposed, and the asymptotic closed-form expression for the system outage performance and the diversity order for SJ-AS were both obtained. Numerical results demonstrated that SJ-AS can outperform both the conventional max-min approach and the random selection scheme, and can achieve the optimal performance of the ES~algorithm.\vspace{-0.5em}
\ifCLASSOPTIONcaptionsoff
  \newpage
\fi


\begin{thebibliography}{1}

\bibitem{Ref_noma_5G}
L. Dai, B. Wang, Y. Yuan, S. Han, I. Chin-Lin, and Z. Wang, ``Non-orthogonal multiple access for 5G: solutions, challenges, opportunities, and future research trends", \emph{IEEE Commun. Mag.}, vol. 53, pp. 74-81, Sep. 2015.

\bibitem{Ref_cognitive_radio}
S. Haykin. ``Cognitive radio: brain-empowered wireless communications",  \emph{IEEE journal on selected areas in communications}, vol. 23, no. 2, pp. 201-220, Feb. 2005.

\bibitem{Ref_CR_NOMA}
Z. Ding, P. Fan, and H. V. Poor, ``Impact of user pairing on 5G non-orthogonal multiple access downlink transmissions", \emph{IEEE Trans. on Veh. Technol.}, vol. 65, no. 8, Aug. 2016.

\bibitem{Ref_CR_NOMA_MIMO}
Z. Ding, R. Schober, and H. V. Poor. ``A general MIMO framework for NOMA downlink and uplink transmission based on signal alignment", \emph{IEEE Trans. on Wireless Commun.}, vol. 15, no. 6, pp: 4438-4454, Mar. 2016.

\bibitem{Ref_ME2}
Y.~Yu, H. Chen, Y. Li, Z. Ding, and B. Vucetic, ``Antenna selection for MIMO non-orthogonal multiple access systems", arXiv: 1612.04943 (2016)

\bibitem{Ref_outage2}
H. Kong, and Asaduzzaman,  ``On the outage behaviour of interference temperature limited CR-MISO channel", \emph{IEEE Journal of Commun. and Networks}, vol. 13, no. 5, pp. 456-462, 2011.

\bibitem{Ref_one_RF1}
E. Erdogan, A. Afana, S. Ikki, and H. Yanikomeroglu, ``Antenna selection in MIMO cognitive AF relay networks with mutual interference and limited feedback¡±, \emph{IEEE Commun. Lett.}, vol. 21, no.5, May 2017.

\bibitem{Ref_power_allocation}
M. Zeng, G. I. Tsiropoulos, O. A. Dobre, and M. H. Ahmed, ``Power allocation for cognitive radio networks employing non-orthogonal multiple access," in \emph{Proc. IEEE Global Commun. Conf.}, Dec. 2016.

\bibitem{Ref_outage}
Z. Ding, M. Peng, and H. V. Poor, ``Cooperative non-orthogonal multiple access in 5G systems", \emph{IEEE Commun. Lett.}, vol. 19, no. 8, Aug. 2015.

\end{thebibliography}
\end{document}